\documentclass[runningheads]{llncs}
\usepackage[T1]{fontenc}
\usepackage{graphicx}
\usepackage{booktabs}
\usepackage[misc]{ifsym}
\newcommand{\corr}{(\Letter)}
\usepackage{mwe}
\usepackage{multirow}
\usepackage[table]{xcolor}
\usepackage{amsfonts}
\usepackage{array}
\usepackage{stfloats}
\usepackage{url}
\usepackage{verbatim}

\def\BibTeX{{\rm B\kern-.05em{\sc i\kern-.025em b}\kern-.08em
 T\kern-.1667em\lower.7ex\hbox{E}\kern-.125emX}}
\usepackage{balance}

\usepackage{amsmath}
\usepackage{cases}
\usepackage[normalem]{ulem}
\usepackage{enumerate}
\usepackage[linesnumbered,ruled]{algorithm2e}

\usepackage[hang,flushmargin]{footmisc}
\usepackage{algpseudocode}
\usepackage{textcomp}
\usepackage{epstopdf}
\usepackage{subcaption}
\usepackage{rotating}
\usepackage{bbm}

\usepackage{newfloat}
\usepackage{listings}
\usepackage{hhline}
\usepackage{bibentry}
\usepackage{amssymb}
\usepackage{dsfont}
\usepackage{mathrsfs}

\usepackage{mathabx}

\newtheorem{thm}{Theorem}
\newtheorem{lem}{Lemma}

\makeatletter

\newcommand{\Rmnum}[1]{\expandafter\@slowromancap\romannumeral #1@}
\makeatother

\begin{document}

\title{Voronoi Diagram Encoded Hashing}

\titlerunning{Voronoi Diagram Encoded Hashing}

\author{Yang Xu \and Kai Ming Ting \corr}
\authorrunning{Y.Xu and K.M.Ting}
\institute{National Key Laboratory for Novel Software Technology, Nanjing University, Nanjing 210023, China \\ \email{xuyang@lamda.nju.edu.cn, tingkm@nju.edu.cn}}

\maketitle              % typeset the header of the contribution

\begin{abstract}
The goal of learning to hash (L2H) is to derive data-dependent hash functions from a given data distribution in order to map data from the input space to a binary coding space. 
Despite the success of L2H, two observations have
cast doubt on the source of the power of L2H, i.e., learning. First,  a recent study shows that even using a version of locality sensitive hashing functions without learning achieves binary representations that have comparable accuracy
as those of L2H, but with less time cost. Second,  existing L2H methods are constrained to three types of hash functions: thresholding, hyperspheres, and hyperplanes only. In this paper, we unveil the potential of Voronoi diagrams in hashing. 
Voronoi diagram is a suitable candidate because of its three properties. This discovery has led us to propose a simple and efficient no-learning binary hashing method, called Voronoi Diagram Encoded Hashing (VDeH), which constructs a set of hash functions through a data-dependent similarity measure and produces independent binary bits through encoded hashing. We demonstrate through experiments on several benchmark datasets that VDeH achieves superior performance and lower computational cost compared to existing state-of-the-art methods under the same bit length.

\keywords{Learning to hash  \and Binary representation \and Voronoi diagram }
\end{abstract}

\section{Introduction}

For decades, hashing techniques have been one of the most effective tools commonly used to compress data for fast access, analysis, and learning~\cite{LSH-similarity-estimation-2002,chi2017hashingsurvey,zhu2023multi}. Hashing techniques are popular for their simplicity and offer significant advantages in data processing and security~\cite{al2020secure,hu2021persistent}. Their primary strength lies in the ability to efficiently map arbitrary-sized input data to fixed-size embeddings, known as hash values. This process is deterministic, meaning the same input always yields the same output, ensuring data consistency. These attributes facilitate a wide range of applications, including data integrity verification, data indexing, blockchain technology, and distributed systems~\cite{borthwick2021scalable,mao2022hash,min2024sephash,thangavel2019enabling}.

Learning to hash (L2H)~\cite{rcLSH23,Wang-02-survey,CBE14}, a prominent subfield within hashing techniques, aims to map data from the input space to a binary coding space, in which the Hamming distance well approximates the distance (e.g., $\ell_p$ norm) in input space. Then, efficient retrieval and learning can be performed in the binary coding space.  The general problem in L2H is to learn long binary codes that approximate the input distance\footnote{Note that another related area called \textit{deep hashing}~\cite{luo2023survey,singh2022learningsurvey} conducts a similar task but with a totally different goal: generating short compact binary codes, where proximate or identical binary codes represent semantic (categorical) similarity between instances, while distant binary codes represent semantic dissimilarity. However, deep hashing is specifically designed for datasets with explicit categorical information, notably image data~\cite{Hansen2020UnsupervisedSHDeep,he2024one,Wu2022OnlineEnhancedDeep}, and is incompatible with datasets lacking this information, including extensive text and tabular data. Besides, these short binary codes do not adequately approximate the input distance. Existing studies demonstrate that achieving an effective approximation of the input distance through binary codes requires a code length of $\mathcal{O}(d)$, where $d$ is the input dimensionality~\cite{BPH13,rcLSH23,SP15,CBE14}.}.
% \textcolor{red}{I think this footnote can be moved to the section on deep learning so that every thing about it is discussed in one section.}}.

Given input data of size $N$ with dimensionality $d$, the core task of L2H is to construct binary hash functions that map the data into $L$-bit binary codes. Broadly, existing L2H methods can be categorized into three classes: thresholding~\cite{BPH13,ITQ13,SHnips08,SP15,CBE14}, hyperspheres~\cite{heo2015spherical}, and hyperplanes~\cite{LSH21,DSH14,rcLSH23,SELVE14}. Thresholding-based methods initiate the process by transforming distance or similarity relationships from the input space into a new matrix of size $N\times L$ through metric learning~\cite{SHnips08,CBE14} or PCA-based methods~\cite{BPH13,ITQ13,SP15}. Then, matrix values are binarized based on a given threshold. Hypersphere and hyperplane-based methods construct hash functions by partitioning the input space. Points falling within the same partition are mapped to identical binary code values. Due to randomness in the partitioning process, these methods typically improve binary codes performance by learning the partitioning strategy to achieve a uniform distribution of points across distinct partitions~\cite{heo2015spherical,rcLSH23}.

Despite the success of L2H, 
the following two observations cast doubt on the source of the power of methods, i.e., learning:

\begin{enumerate}
    \item Existing L2H methods all claim that their learning process (in optimizing some defined objective) plays a crucial role in obtaining effective binary codes. Yet, a recent study \cite{LSH21} shows that even using brLSH, a version of locality-sensitive hashing~\cite{LSH-similarity-estimation-2002} without learning, achieves binary codes that have comparable performance as those of L2H methods, but with less computation cost. 
   \item They are constrained to three types of hash functions. Learning is required in order to ensure that these functions satisfy three key properties for effective binary codes~\cite{ITQ13,heo2015spherical,DSH14,rcLSH23,SHnips08,SP15}, i.e., (a) \textit{full space coverage}: each hash function covers the entire space; (b) \textit{entropy maximization}: each hash function covers approximate the same number of points from the input data, such that the total set of hash functions maximizes information entropy; and (c) \textit{bit independence}: all hash functions have mutually independent hash bits. Their importance will be discussed in detail in Section \textbf{3}.
\end{enumerate}

In this work, we provide the insight that Voronoi diagram is a suitable alternative to L2H, attributed to its three properties: first, a Voronoi diagram naturally covers the entire input space; second, for a given dataset, every Voronoi cell in a Voronoi diagram covers approximately the same number of points \cite{Devroye2017OnTM}; third, it is feasible to transform a Voronoi diagram into a set of hash functions, which generates mutually independent bits (hash values). Unlike existing L2H methods, all these properties are obtained in Voronoi diagrams without any extra learning process.

Voronoi diagrams facilitate the generation of \emph{data-dependent} hash functions, as these hash functions are derived by samples from a given dataset and are dependent on the distribution of the data. Specifically, they yield small Voronoi cells in areas of high densities, while large Voronoi cells are produced in areas of low densities \cite{Devroye2017OnTM,ting2018isolation}. 
This discovery has led us to propose a simple and effective no-learning approach called Voronoi Diagram Encoded Hashing (VDeH), which is the first data-dependent binary hashing scheme based on Voronoi diagram already used in a kernel. 
Figure \ref{fig:Voronoi} shows an illustrative comparison of VDeH using Voronoi diagram and three existing types of hash functions.

\begin{figure*}[t!]
    \centering
    \includegraphics[width=0.93\textwidth]{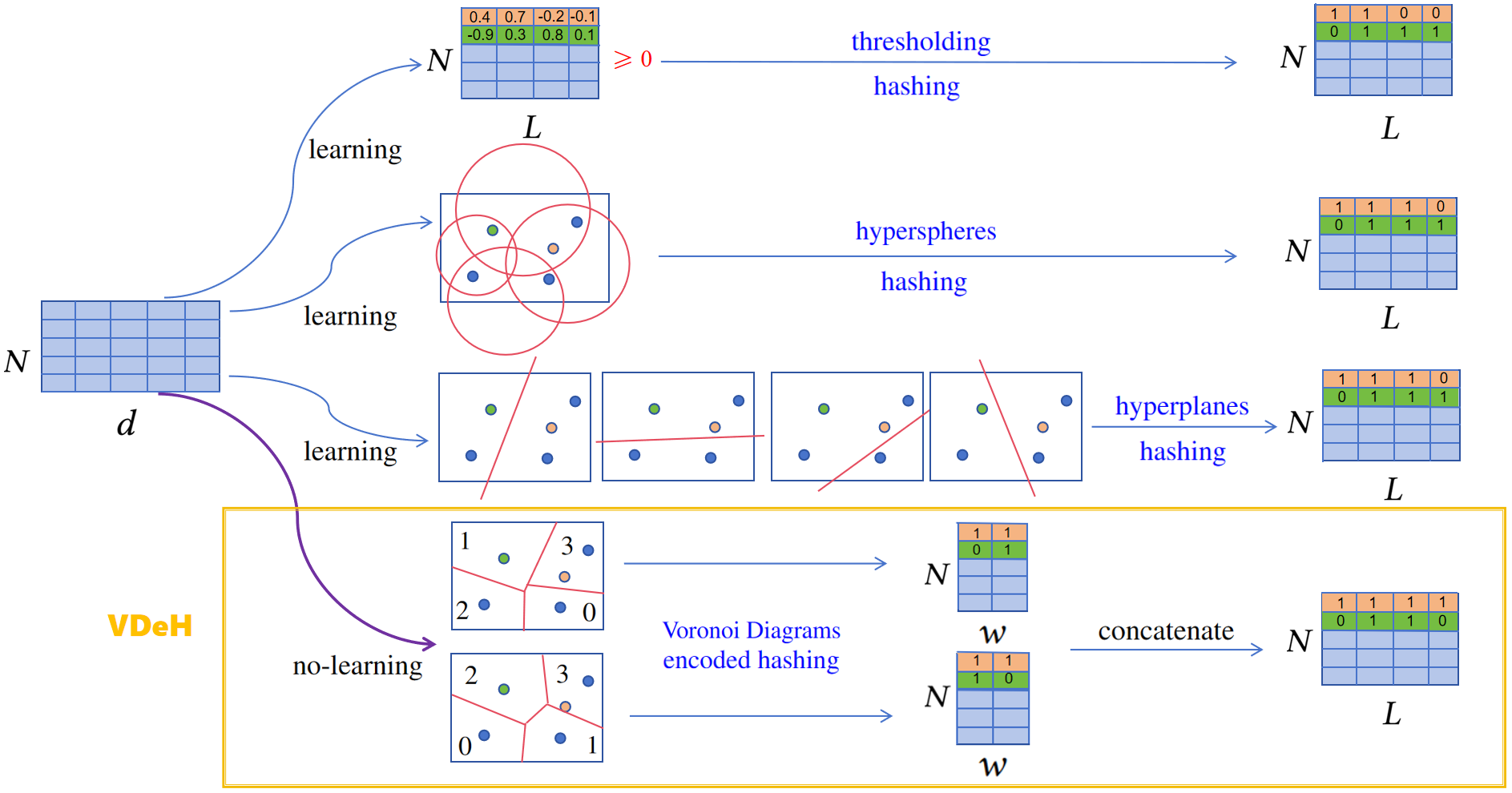}
    \caption{An illustrative comparison of VDeH using Voronoi diagram and three existing types of hash functions (thresholding, hyperspheres and hyperplanes). The parameters $N=5$, $d=5$, $L=4$, and $w=2$ are used in this example. 
    }
    \label{fig:Voronoi}
\end{figure*}

Our main contributions are summarized as follows:

\begin{itemize}
    \item Providing the insight that Voronoi diagram can be leveraged to realize a family of data-dependent hash functions, without learning. 
    \item Creating a new definition of hashing scheme derived directly from  a kernel based on Voronoi diagrams.
    \item Proposing VDeH, which creates the first kernel-based implementation of a data-dependent binary hashing scheme, and formulating a binary distance function intrinsic to VDeH. We have theoretically proved that VDeH is capable of generating mutually independent binary bits.
    \item Conducting comprehensive empirical comparisons and analyses between VDeH and existing state-of-the-art L2H methods to demonstrate the effectiveness and efficiency of VDeH.
\end{itemize}

\section{Background and Related Work}

With the emergence of large-scale data,  numerous learning to hash (L2H) methods have been developed~\cite{haq2021survey,rcLSH23,Wang-02-survey}. The generic binary hashing problem is the following. Given $N$ data points $\mathcal{X}=[\mathbf{x}_1,\dots,\mathbf{x}_N]\in\mathbb{R}^{d\times N}$, generate $L$ hash functions to map a data point $\mathbf{x}$ into a $L$-bit binary hash code
$\mathcal{B}(\mathbf{x})=[h_1(\mathbf{x}),h_2(\mathbf{x}),\dots,h_L(\mathbf{x})]$
where $h_l(\mathbf{x})\in\{0,1\}$ is the $l$-th hash function. 

For the
linear binary projection-based hashing~\cite{BPH13,rcLSH23,CBE14}:
$$h_l(\mathbf{x})= sgn(F(\mathbf{R}_l\mathbf{x}+t_l)),$$
where $\mathbf{R}_l\in\mathbb{R}^{L\times d}$ is the projection matrix, $sgn(\cdot)$ is a binary map, and $t_l$ is the intercept. Different hashing methods aim at finding different $F$, $\mathbf{R}_l$ and $t_l$ with respect to different objective functions.

Existing L2H methods seek to ensure that the Hamming distance of generated binary codes closely approximates the input distance by imposing specific constraints on these codes. For example, given $b_i \in \{0, 1\}^L$, Spectral Hashing (SH)~\cite{SHnips08} aims to minimize the average Hamming distance between similar neighbors, represented by $\sum_{ij} W_{ij} ||b_i - b_j||^2$, where $W_{ij}$ is the similarity between points $i$ and $j$ as measured by Gaussian kernel, and $||b_i - b_j||^2$ is the Hamming distance between binary codes $b_i$ and $b_j$. SH uses two constraints: (i) 
$\sum_i b_i = \frac{1}{2} $ to ensure entropy maximization, and (ii) $\frac{1}{n} \sum_i b_i b_i^T = I$ to guarantee bit independence.  
%Here, $\sum_i b_i = 1/2$ ensures bit balance, and $\frac{1}{n} \sum_i b_i b_i^T = I$ ensures bits are mutually independent. 
In addition, Spherical Hashing (SpH)~\cite{heo2015spherical} generates hash functions by partitioning the input space using hyperspheres. SpH achieves bit independence by ensuring that each hashing function has an equal probability of output 0 and 1. To achieve entropy maximization, it requires each hypersphere to contain an approximately equal number of points. Since hyperspheres do not cover the entire input space, SpH increases the radius of the hyperspheres to cover at least all training data, thus ensuring full space coverage. 
It is worth noting that most existing methods explicitly or implicitly ensure the three key properties: full space coverage, entropy maximization, and bit independence by constructing similar constraints. Methods operated on spatial partitioning necessitate that the hash functions provide complete coverage of the input space~\cite{LSH21,heo2015spherical,DSH14,rcLSH23,SELVE14}. Besides, existing methods ensure entropy maximization by explicitly constraining each bit to have a half probability of being 0 or 1, while maintaining the independence of each bit's output~\cite{BPH13,ITQ13,heo2015spherical,DSH14,rcLSH23,SHnips08,SP15,CBE14,SELVE14}. The detailed explanations of how these methods guarantee these properties and their categorizations can be found in comprehensive survey papers~\cite{luo2023survey,Wang-02-survey}.

The three properties have been demonstrated to be crucial for effective binary codes in existing methods. However, these methods invariably require a learning process to gain these properties, as their employed hash functions, including those based on thresholding, hyperspheres, and hyperplanes, do not naturally possess these properties.

\section{Insight: Voronoi Diagram is a Suitable Candidate for Binary Hashing}
\label{sec-insight}

The key notations used in this paper are provided in Table \ref{notations}.

\begin{table}[t!]
\renewcommand{\arraystretch}{1}
\centering
% \caption{Key notations used. `VD' denotes Voronoi diagram.}
\caption{Key notations used.}
\label{notations}
\begin{tabular}{@{}c|l@{}}
\hline
Notation & Definition \\ \hline
$\mathcal{X}$ & Input dataset with $N$ points in $\mathbb{R}^d$ \\
$\mathcal{D}$ & A set of $\psi$ points randomly sampled from $\mathcal{X}$ to generate a Voronoi diagram \\
% $\psi$ & The number of points in $\mathcal{D}$ \\
$\mathcal{B}$ & $L$-bit binary codes correspond to all points in $\mathcal{X}$ \\
% $L$ & The number of bits of a binary code in $\mathcal{B}$ \\
$H$ & A set of hash functions derived from a Voronoi Daigram \\
$\kappa$ & kernel $\kappa(x,y)$, where $x,y \in \mathbb{R}^d$\\
$\mathbf{S}$ & A similarity function defined in $\mathbb{R}^d$ \\
$w$ & The number of bits generated by encoded hashing over a Voronoi diagram \\
$\mathcal{H}$ & A hashing scheme consisting of a family of hash functions. \\
$P_\mathcal{H}$ & Distribution over a family of hash functions $\mathcal{H}$ \\
\hline
\end{tabular}
\end{table}

Given a point $\mathbf{x} \in \mathbb{R}^d$, a hashing scheme of a family $H$ of hash functions $h \in H$, where $h(\mathbf{x}) \in \{0,1\}$, must have the following three properties: full space coverage, entropy maximization, and bit independence. 

\textbf{Full space coverage.} All $h \in H$ cover the entire $\mathbb{R}^d$, i.e.,
$\forall \mathbf{x} \in \mathbb{R}^d, \exists h \in H \mbox{ s.t. } h(\mathbf{x}) = 1.$
%Otherwise, those uncovered points are not retrievable. 
The failure to achieve this coverage can severely impair retrieval effectiveness. Two possible current treatments of this shortcoming are unsatisfactory: (a) All points outside the covered regions have no binary code representations, rendering them irretrievable. (b) Using a common binary mapping to all points outside the covered regions risks conflating vastly distant points in the input space with a same binary code. This issue is particularly pronounced when adopting hypersphere-based hashing strategies \cite{heo2015spherical}. 
Therefore, ensuring that hash functions cover the entire space is  a critical aspect of maintaining high retrieval accuracy  in binary hashing-based information retrieval systems. This coverage guarantees that every point $\mathbf{x} \in \mathbb{R}^d$ can be effectively indexed and retrieved, thereby maximizing the utility of binary hashing in real-world applications.

\textbf{Entropy maximization.} For a given dataset $\mathcal{X} \subset \mathbb{R}^d$ with $N$ points, every hash function $h \in H$ shall cover approximately the same number of points $x \in \mathcal{X}$, i.e., uniformly distributed over all hash functions:
$$ \forall h \in H, \left| \textstyle|h(\mathcal{X})| - \frac{N}{|H|} \right| \leq \varepsilon ,$$
where $|H|$ denotes the total number of hash functions in $H$ and $\varepsilon$ is a predefined non-negative small number indicating the tolerance level for the approximation. This is to ensure that there are no under-utilized hash functions. This uniform coverage can be understood as maximizing the information entropy of the resulting binary codes \cite{Li2021DeepUI}, where high information entropy is indicative of a rich, diverse representation of the dataset.
It avoids the scenario where a large proportion of the dataset is lumped in a few binary codes only.

\textbf{Bit independence.} All hash functions in $H$ are mutually independent. This means that for any pair of distinct hash functions $h_i \ne h_j \in H$, and for any point $\mathbf{x} \in \mathbb{R}^d$, the outputs $h_i(\mathbf{x})$ and $h_j(\mathbf{x})$ are statistically independent. It has been proven to be indispensable for optimizing performance \cite{He2011CompactHW,heo2015spherical,JolyBuisson2011RMMH,SHnips08}. This independence ensures that each hash function contributes uniquely to the hashing process, thereby eliminating potential redundancy in the generated binary bits. When hash functions are not mutually independent, the resulting binary representations suffer from information redundancy, which in turn, dilutes the effectiveness of the binary hashing scheme, leading to poor retrieval efficiency. Such redundancy also inflates the storage requirements unnecessarily.

It is interesting to note that none of the existing hash functions (i.e., thresholding, hyperspheres and hyperplanes) have the three properties naturally. That is the reason why learning has been employed to ensure that the final hash functions satisfy the three properties.

Here we have the insight that Voronoi diagrams is a suitable candidate for hash functions because they have the following properties, without learning: 
\begin{enumerate}[1.]
    \item A Voronoi diagram naturally covers the entire input space.
    \item For a given dataset, every Voronoi cell in a Voronoi diagram covers approximately the same number of points. 
    This occurs when a set of points, that represents the data distribution, is used to construct the Vonoroi diagram.  
    And it can be easily achieved through a random Voronoi partition created by a set of  points drawn independently and randomly from the dataset \cite{Devroye2017OnTM}.
    \item It is feasible to transform a Voronoi diagram into a set of mutually independent hash functions (see the analysis in Section \textbf{4.3}).
\end{enumerate}

This insight has led us to propose a no-learning data-dependent hashing scheme based on Voronoi diagrams, described in the next section.

\section{Proposed Approach}

Here we give the formal definition of our proposed binary hashing method called Voronoi Diagram Encoded Hashing (VDeH), which is the first binary hashing scheme based on a data-dependent similarity using Voronoi diagram partitioning.

\subsection{A New Definition of Hashing}

 A kernel based on Voronoi diagrams called Isolation Kernel \cite{ting2018isolation} is defined as:
\[
\kappa(\mathbf{x},\mathbf{y} |\mathcal{H}(\mathcal{X}))  =  {\mathbb E}_{H \sim \mathcal{H}(\mathcal{X})} [\mathds{1}({\mathbf{x}},{\mathbf{y}} \in \theta\ | \ \theta \in H)],\]
\noindent where each Voronoi diagram $H$ has cells $\theta$, and $\mathds{1}(\cdot)$ is an indicator function.

By re-writing each cell $\theta$ as a hash function $h$, it can be re-expressed as:
\[
\kappa(\mathbf{x},\mathbf{y} |\mathcal{H}(\mathcal{X}))    =  {\mathbb E}_{H \sim \mathcal{H}(\mathcal{X})} [\mathds{1}(h(\mathbf{x}) = h(\mathbf{y}) = 1\ | \ h \in H)].
\]
The above revelation, together with the three properties of Voronoi diagram (stated in Section \ref{sec-insight}), prompt us to propose a new definition of hashing.
Given a dataset $\mathcal{X} \subset \mathbb{R}^{d\times N}$, the proposed VDeH scheme is defined as follows: 

\begin{definition}
    The VDeH scheme is a family $\mathcal{H}(\mathcal{X})$ of  hash functions  created by Voronoi diagrams associated with a distribution $P_\mathcal{H}$ over $\mathcal{H}(\mathcal{X})$ generated from a dataset $\mathcal{X}$  such that it satisfies
    %such that a function $H \in \mathcal{H}(\mathcal{X})$ selected based on $P_\mathcal{H}$ satisfies
%     \[
%     Pr_{H \in H(\mathcal{X})}[h(x) = h(y)] = \kappa(x,y |H(\mathcal{X}))
% \]
\begin{equation}
    Pr_{h\in H \in \mathcal{H}(\mathcal{X})}[h(\mathbf{x}) = h(\mathbf{y})=1] = \kappa(\mathbf{x},\mathbf{y} |\mathcal{H}(\mathcal{X})),
\label{eq:lshour}
\end{equation}
\noindent
where $\kappa(\mathbf{x},\mathbf{y} |\mathcal{H}(\mathcal{X}))$
%\approx \displaystyle \frac{1}{|H||F|}\sum_{H \in F} \mathds{1}[h(x) = h(y) = 1]_{h \in H}$ 
is derived from Voronoi diagrams generated from $\mathcal{X}$; 
%$F$ is a finite set of $\mathcal{H}(\mathcal{X})$ sampled according to $P_\mathcal{H}$;
and $H$ is a set of all hash functions $h$ derived from a Voronoi diagram.
\label{def1}
\end{definition}

This definition makes the implementation of VDeH extremely simple because the hashing functions can be obtained with a simple additional encoding from the Voronoi diagrams already used in Isolation Kernel. Unlike existing hashing methods, VDeH needs no special design or learning of hash functions to gain the three properties mentioned in Section \ref{sec-insight}. 
% Its implementation details are given in the next section.

\subsection{Implementation Details of VDeH}
\label{sec-implementation}

Given a subset $\mathcal{D} = \{\mathbf{s}_1,\dots,\mathbf{s}_\psi\} \subset \mathcal{X}$ with $\psi$ randomly selected points.
% Given a set $\mathcal{D}$ of randomly selected points $s_i \in \mathbb{R}^{d}$ \textcolor{red}{Rewrite to: Given a subset $\mathcal{D} \subset D$ (also replace $\mathcal{D}$ with $\mathcal{D}$ throughout. It is weird to have this set denoted using a math boldface small case letter rather than a capital letter.} of randomly selected points $s_i, i=1,\dots,m$. 
A Voronoi diagram $H$, created by $\mathcal{D}$, partitions the $\mathbb{R}^{d}$ space into $\psi$ Voronoi cells, where $\mathbf{s}_i$ is at the center of Voronoi cell $i$. Let a set of hash functions created by $\mathcal{D}$ be $H =  \{ h_1,h_2,\dots,h_{\psi} \}$. For any $\mathbf{x}\in \mathbb{R}^{d}$,  $h_i(\mathbf{x})$ is defined as:
% $$h_i(x)=
% \begin{cases}
%     0 & if \quad dist(x,s_i) > dist(x,\mathcal{D})\\
%     1 & if \quad dist(x,s_i) = dist(x,\mathcal{D}) 
% \end{cases}
% $$
\begin{equation}
    h_i(\mathbf{x})=
    \begin{cases}
    0 & \text{if} \quad dist(\mathbf{x},\mathbf{s}_i) > dist(\mathbf{x},\mathcal{D})\\
    1 & \text{if} \quad dist(\mathbf{x},\mathbf{s}_i) = dist(\mathbf{x},\mathcal{D}), 
    \end{cases}
\label{eqhash}
\end{equation}
% \textcolor{red}{How do you deal with the boundary points which can belong to more than one regions?}
where $dist(\mathbf{x},\mathbf{s}_i)= \parallel \mathbf{x}-\mathbf{s}_i \parallel$ denotes the $\ell_2$ distance between $\mathbf{x}$ and $\mathbf{s}_i$; and $dist(\mathbf{x},\mathcal{D}) = \min_{\mathbf{y} \in \mathcal{D}/\{\mathbf{x}\}}  \parallel \mathbf{x}-\mathbf{y} \parallel$. When a point is located on a boundary, it is randomly assigned to any one of the cells sharing that boundary.

Given a fixed length of binary bits,
existing studies have established that ensuring independence among generated binary bits is key to effective binary hashing. However, directly converting Voronoi cells into hash functions results in dependencies between them, due to the fact:

\begin{equation}
    \forall h\in H, \forall \mathbf{x}\in \mathbb{R}^{d}, \sum_{i\in [1,\psi]} \mathds{1}[h_i(\mathbf{x}) = 1] = 1.
\label{eq:character}
\end{equation}
% $$ \forall h\in F, \forall x\in \mathbb{R}^{d}, \sum \mathds{1}[h_i(x) = 1]_{i\in [1,m]} \equiv 1 $$
% \textcolor{red}{It is unclear how this leads to dependent binary bits}

This dependency arises because if $\mathbf{x}$ belongs to the $i$-th Voronoi cell, indicating $h_i(\mathbf{x})=1$, then $\mathbf{x}$ can not fall into other regions, hence $h(\mathbf{x})=0$ for those. As a result, any $h(\mathbf{x})$ is influenced by other hash functions.

To address this issue, we have adopted a simple encoded hashing mechanism to generate mutually independent binary bits. Let $\psi=2^w$, and as $2^w$ Voronoi cells can be encoded with a binary code of $w$ mutually independent bits, a $L$-bit code of the VDeH scheme represents $L/w \times {2^w}$ hash functions (the proof is presented in the Section \textbf{4.3}). 
% This property of the encoded hashing mechanism is derived from the following theorem.

The process for generating a binary code corresponding to a point $\mathbf{x} \in \mathbb{R}^d$ via VDeH is outlined as follows:

\begin{enumerate}[1.]
    \item Given a dataset $\mathcal{X}$, we randomly sample $\psi$ points to form a set $\mathcal{D}$, and generate the set of Voronoi diagram hash functions $H =  \{ h_1,h_2,\dots,h_{\psi} \}$. According to Eq.(\ref{eq:character}), there is  one and only one $h_i(\mathbf{x})=1$ for any point $\mathbf{x}$, and $\forall u \ne i, h_u(\mathbf{x})=0$.
    \item 
    %Without loss of generality, considering the $i$-th hash value as 1, we map it to 
    The above `raw' $\psi$-bit vector $[h_1(\mathbf{x}),h_2(\mathbf{x}),\dots,h_\psi(\mathbf{x})]$ can then be encoded as $w$-bit vector $b(\mathbf{x}) = [e_1(\mathbf{x}),e_2(\mathbf{x}),\dots,e_w(\mathbf{x})]$, where $\psi \le 2^w$,  through the following encoded hashing function:
    \begin{equation}
        e_j(\mathbf{x}) = 
\begin{cases}
    % 0 & if \quad i\mod{2^{j-1}} = 0 or i < 2^{j-1}\\
    % 1 & if \quad i\mod{2^{j-1}} \neq 0
    0 & \text{if} \quad \lfloor\frac{\arg_{i\in[1,\psi]} h_i(\mathbf{x}) = 1}{2^{j-1}}\rfloor \mod 2 = 0 \\
    1 & \text{if} \quad \lfloor\frac{\arg_{i\in[1,\psi]} h_i(\mathbf{x}) = 1}{2^{j-1}}\rfloor \mod 2 \neq 0, \\
\end{cases} 
\label{eqhash2}
    \end{equation}
    where $j = 1,2,\dots,w$.
    \item After repeating the above two steps $L/w$ times, we concatenate all $w$-bit vectors $b(\mathbf{x})$ to formulate a $L$-bit code $\mathcal{B}(\mathbf{x})=[b^1(\mathbf{x}),\dots,b^{L/w}(\mathbf{x})]$ for any point $\mathbf{x}$.

\end{enumerate}

It is interesting to note that this simple encoding is not applicable to existing L2H methods where a $L$-bit code could represent $L$ hash functions only. In short, the VDeH scheme represents $(\frac{\psi}{w} -1)L$ more hash functions than existing L2H methods when both have codes of the same number of $L$ bits. This is a key to VDeH's better performance over existing L2H methods.
Also note that the Voronoi diagram does not necessarily have to have $2^w$ Voronoi cells. It is simply a convenience for a binary encoded hashing with mutually independent bits.

\textbf{The time complexity of  VDeH}. VDeH has a linear time complexity with respect to the size of the dataset, denoted as $N$, and the dimensionality of the data, denoted as $d$. A detailed analysis of the time complexity of VDeH is presented as follows: to convert a dataset $\mathcal{X}$ of $N$ points of $d$ dimensions, VDeH first finds the nearest neighbor of each point in $\mathcal{X}$ from the set $\mathcal{D}$ of $\psi$ points (or Voronoi cells), resulting in a complexity of $\mathcal{O}(\psi Nd)$. This process is repeated $L/w$ times, leading to a total complexity of $\mathcal{O}(\frac{\psi LNd}{w})$. VDeH then uses the encoded hash functions $h$ to generate $L$-bit binary codes for all $N$ points, resulting in a complexity of $\mathcal{O}(LN)$. Thus, the overall time complexity for VDeH's `training' process is $\mathcal{O}(\frac{\psi LNd}{w})$, which is linear to the dataset size $N$ and the input dimensionality $d$.

\subsection{Distance Measure for VDeH Codes}

Hamming distance, renowned for its simplicity and efficiency, is predominantly employed in the comparison of binary codes due to its methodological advantage of merely counting the differing bits between two codes, thereby circumventing the necessity for intricate arithmetic operations such as multiplication and square root calculations. This stands in stark contrast to the computationally more demanding Euclidean distance.

For VDeH, the measurement of similarity between any two points is conducted by calculating the probability that both points fall into the same regions across all generated Voronoi cells. 
To fully utilize the three properties of the Voronoi diagram hashing functions, we propose the following distance metric, VDeH distance ($\mathbf{d}_{V}(\mathcal{B}_i,\mathcal{B}_j)$), between two binary codes $\mathcal{B}_i$ and $\mathcal{B}_j$ generated by VDeH:
% Thus, 
% a distance measure for VDeH codes that directly computes an effective distance equals to the abovementioned probability between $x$ and $y$ through $\mathcal{B}(\mathbf{x})$ and $B(y)$, as the following:
\begin{equation}
    \mathbf{d}_{V}(\mathcal{B}_i,\mathcal{B}_j) =\frac{w}{L} { \sum_{k=1}^{L/w} \mathds{1}(b_{i}^k \neq b_j^k)},
\label{eq:distance}
\end{equation}
where $\mathcal{B}_i=[b_i^1,\dots,b_i^{L/w}]$ and $\mathcal{B}_j=[b_j^1,\dots,b_j^{L/w}]$.
Then, for any given query point $\mathbf{q}$, after computing its binary code $\mathcal{B}(\mathbf{q})$ through VDeH and its distance to every point in the dataset can be calculated using Eq.(\ref{eq:distance}), the retrieval process returns a point in the dataset having the shortest distance to the query point $\mathbf{q}$.

\subsection{Independence among VDeH Bits}

The importance of independence among hash bits are mentioned in existing studies \cite{SHnips08,He2011CompactHW,JolyBuisson2011RMMH}. It can be defined as follows:

\begin{definition}
    (\textbf{Independence among hash bits.}) Given a family of hash functions $\{h_1,h_2,\dots,h_L\}$, for any $\mathbf{x} \in \mathbb{R}^d$,  $h_1(\mathbf{x})$, $h_2(\mathbf{x})$, $\dots$, $h_L(\mathbf{x})$ are said to be mutually independent for $\forall \mathbf{x}$ if for any set of values $v_1,v_2,\dots,v_L \in \{0,1\}$, it holds that:
    \begin{align*}
        Pr(h_{1}(\mathbf{x})=v_{1},h_{2}(\mathbf{x})=v_{2},\dots,h_{L}(\mathbf{x})=v_{L}) = \\
        Pr(h_{1}(\mathbf{x})=v_{1})\cdot Pr(h_{2}(\mathbf{x})=v_{2})\cdot \dots \cdot Pr(h_{L}(\mathbf{x})=v_{L}).
    \end{align*}
    % $$P(h_{i_1}(x)=x_{i_1},h_{i_2}(x)=x_{i_2},\dots,h_{i_k}(x)=x_{i_k}) =$$
    % $$ P(h_{i_1}(x)=x_{i_1})\cdot P(h_{i_2}(x)=x_{i_2})\cdot \dots \cdot P(h_{i_k}(x)=x_{i_k})$$
\end{definition}

This means that the value of one hash bit does not influence the value of another hash bit.
Although directly converting regions generated by Voronoi diagrams into hash functions results in dependencies between them based on the Eq.(\ref{eq:character}), we prove that the encoded hashing enables the achievement of such independence among hash bits. VDeH generates mutually independent binary bits through encoded hashing, as demonstrated by the following theorem.

\begin{lem} \label{lem2}

Let $\mathcal{D} = \{\mathbf{s}_1,\mathbf{s}_2,\dots,\mathbf{s}_{\psi}\} \thicksim G^{\psi}$ be a dataset, where every $\mathbf{s}_i$ is i.i.d drawn from an unknown probability distribution $G$ on the input space $\Omega$. Let $\mathcal{D}$ forms its Voronoi cells $V_i \subset \Omega$, the probability that $\mathbf{s}_i$ is the nearest neighbor of any $\mathbf{x}\in \Omega$ in $\mathcal{D}$ is given as: $Pr(\mathbf{x}\in V_i)=1/{\psi}$, for every $i=1,2,\dots,{\psi}$.

\end{lem}

\begin{proof}

Let $R({\mathbf{x}})=\{{\mathbf{y}} \in \Omega \; | \; dist({\mathbf{x}},{\mathbf{y}}) \leq r \}$ be $r$-neighborhood region of $\mathbf{x}$, and $\Delta R({\mathbf{x}})=\{{\mathbf{y}} \in \Omega \; | \; r \leq dist({\mathbf{x}},{\mathbf{y}}) \leq r+ \Delta r\}$ be $\Delta r$ incremental $r$-neighborhood region of $x$. Let $\rho_G$ be the probability density of $G$, for $\Delta r > 0$, $f = \int_{R({\mathbf{x}})} \rho_G({\mathbf{y}})d{\mathbf{y}}$ and $\Delta f = \int_{\Delta R({\mathbf{x}})} \rho_G({\mathbf{y}})d{\mathbf{y}}$. 

Then, let two events $T$ and $Q$ be as follows:
\begin{eqnarray}
T &\equiv& \mbox{$\mathbf{s}_k \notin R({\mathbf{x}})$ for all $\mathbf{s}_k \in \mathcal{D}$ ($k=1,2,\dots,{\psi}$)},  \mbox{ and}\nonumber\\ 
Q &\equiv& \left\{ \begin{array}{l} 
\mbox{$\mathbf{s}_i \in \Delta R({\mathbf{x}})$ for $\mathbf{s}_i \in \mathcal{D}$},\mbox{ and}\\
\mbox{$\mathbf{s}_k \notin \Delta R({\mathbf{x}})$ for all $\mathbf{s}_k \in \mathcal{D} \; (k=1,2,\dots,{\psi}, i \neq k)$.}
\end{array} \right. \nonumber
\end{eqnarray}

Next, the probability $Pr(T \land Q)=Pr(T)Pr(Q \;|\; T)$ denotes that $\mathbf{s}_i$
is the nearest neighbor of ${\mathbf{x}}$ in $\mathcal{D}$, {\it i.e.}, ${\mathbf{x}} \in V_i$, where
\begin{eqnarray}
Pr(T) &=& (1-f)^{\psi}, \mbox{ and}\nonumber\\
Pr(Q \;|\; T) &= & \frac{\Delta f}{1-f} \left\{1-\frac{\Delta f}{1-f} \right\}^{{\psi}-1}.\nonumber
\end{eqnarray}
By letting $\Delta f$ be infinite decimal $df$, $\Delta f/(1-f) \rightarrow df/(1-f)$ 
and $\{1-\Delta f/(1-f)\} \rightarrow 1$. Thus, we obtain the following 
total probability that $\mathbf{s}_i$ is the nearest neighbor of ${\mathbf{x}}$ in $\mathcal{D}$, {\it i.e.}, ${\mathbf{x}} \in V_i$, by integrating $Pr(T \land Q)$ on $f \in [0,1]$ for every $i=1,2,\dots,{\psi}$.
\[
Pr({\mathbf{x}} \in V_i) = \int_0^1 (1-f)^{\psi} \cfrac{df}{1-f} = \cfrac{1}{{\psi}}\ .
\] 
\end{proof}

\begin{thm} \label{thm3}

Let a $L$-bit code vector $\mathcal{B}(\mathbf{x}) = [b_1(\mathbf{x}),b_2(\mathbf{x}),\dots,$ $b_T(\mathbf{x})]$ denote the VDeH encoded binary vector for any point $\mathbf{x}$, where $T = L/w$ is the number of Voronoi diagrams generated and every $b_i(\mathbf{x})$ is the $w$-bit vector generated by the $i$-th VDeH encoded hashing. When each Voronoi diagram has $2^w$ regions, every bit in $\mathcal{B}(\mathbf{x})$ is mutually independent.

\end{thm}

\begin{proof}
    For any two bits $\alpha$, $\beta \in \mathcal{B}(\mathbf{x})$ located at 
 different Voronoi cells, there are the following two cases: \\
    \textbf{Case 1}: $\alpha$ and $\beta$ are generated by the different encoded hashing. \\
    Let $\alpha\in b_{j}$ and $\beta \in b_{k}$  ($1 \leq j < k \leq T$). Given that Voronoi diagrams generated in each random sampling are mutually independent, the bits obtained from each encoded hashing are also mutually independent. Consequently, it follows that:
    \begin{equation*}
        Pr(b_{j}=s_{j},b_{k}=s_{k}) = Pr(b_{j}=s_{j})\cdot Pr(b_{k}=s_{k}),
    \end{equation*}
    where $s_{j},s_{k} \in \{0,1\}^w$. Moreover, for the bits $\alpha\in b_{j}$ and $\beta \in b_{k}$  
    \begin{equation*}
        Pr(\alpha=v_{\alpha},\beta=v_{\beta}) = Pr(\alpha=v_{\alpha})\cdot Pr(\beta=v_{\beta}),
    \end{equation*}
    where $v_{\alpha},v_{\beta} \in \{0,1\}$. \\
    \textbf{Case 2}: $\alpha$ and $\beta$ are generated by the same encoded hashing. \\
    Let $\alpha,\beta \in b_i$ $(1 \leq i \leq T), b_i = [e_1,e_2,\dots,e_w]$, and $\mathcal{B}$ denote a set of $\psi$ different $w$-bit binary codes correspond to the $\psi$ distinct Voronoi cells. \\
    % When $m = 2^w$, that is, $|\mathcal{B}|= 2^w$, by the pigeonhole principle, then for every possible $w$-bits binary code, there exists a distinct representation in $\mathcal{B}$. Hence, $\mathcal{B}$ exhaustively represents all $w$-bits binary codes.
    When $\psi = 2^w$, $\mathcal{B}$ exhaustively represents all probable $w$-bit binary codes corresponding to $\psi$ Voronoi cells. Let $v_{\alpha},v_{\beta} \in \{0,1\}$, then we have $\sum \mathds{1}[\alpha=v_{\alpha}] = \frac{\psi}{2}$, and $\sum \mathds{1}[\alpha=v_{\alpha},\beta=v_{\beta}] = \frac{\psi}{4}$.
    % \begin{align*}
    %     \sum_{i=1}^m \mathds{1}[\alpha=v_{\alpha}] = \frac{\psi}{2} \\ \sum_{i=1}^m \mathds{1}[\alpha=v_{\alpha},\beta=v_{\beta}] = \frac{\psi}{4}
    % \end{align*}
    Based on Lemma \ref{lem2},  it follows that
    \begin{align*}
         Pr(\alpha=v_{\alpha},\beta=v_{\beta}) & = \textstyle\sum_{i=1}^{\psi}(Pr(\mathbf{x}\in V_i)\mathds{1}[\alpha=v_{\alpha},\beta=v_{\beta}]) \\
        % & = \sum_{i=1}^{\psi}(\frac{1}{\psi}\mathds{1}[\alpha=v_{\alpha},\beta=v_{\beta}]) \\
        & = \textstyle\frac{1}{\psi}\sum_{i=1}^{\psi}\mathds{1}[\alpha=v_{\alpha},\beta=v_{\beta}] \\
        % & = \frac{1}{\psi} (\frac{(\sum_{i=1}^{\psi}\mathds{1}[\alpha=v_{\alpha}])^2}{\psi}) \\
        & = \textstyle \frac{\sum_{i=1}^{\psi}\mathds{1}[\alpha=v_{\alpha}]}{\psi} \times \frac{\sum_{i=1}^{\psi}\mathds{1}[\beta=v_{\beta}]}{\psi} \\
        & = \textstyle\sum_{i=1}^{\psi}(Pr(\mathbf{x}\in V_i)\mathds{1}[\alpha=v_{\alpha}]) \times \sum_{i=1}^{\psi}(Pr(\mathbf{x}\in V_i)\mathds{1}[\beta=v_{\beta}]) \\
        & = Pr(\alpha=v_{\alpha})\cdot Pr(\beta=v_{\beta}).
    \end{align*}
    Consequently, in all cases, $\alpha$ and $\beta$ are mutually independent. Since $\alpha$ and $\beta$ are arbitrary bits in $\mathcal{B}$, when each Voronoi diagram has $2^w$ regions, every bit in $\mathcal{B}(\mathbf{x})$ generated by VDeH is mutually independent.
\end{proof}

\section{Experiment}

\subsection{Datasets and Settings}

To evaluate the proposed Voronoi diagram encoded hashing, we conducted experiments on four public datasets, including two image datasets and two text datasets. It is important to note that existing studies mainly evaluate their performance on image datasets~\cite{BPH13,ITQ13,heo2015spherical,rcLSH23,SHnips08,CBE14,SP15}, our work extends the evaluation to text datasets for a more comprehensive analysis. Unlike images, which have
obviously discriminating features and can be easily identified to guide retrieval tasks, text data presents a greater challenge due to its inherent complexity. The datasets vary in size and dimensionality: CIFAR-10~\cite{cifar10} contains 60,000 images (512 dimensions), GIST~\cite{Gist} contains one million images represented by 960-dimensional descriptors, Nytimes~\cite{nytimes} includes 290,000 articles (256 dimensions), and Kosarak~\cite{kosarak} includes 74,962 click-stream news (27,983 dimensions).

Following existing studies~\cite{heo2015spherical,DSH14,rcLSH23}, we use 10,000 randomly sampled instances for training. We then randomly sample 500 instances, different from the training set as queries.
% we use randomly sampled instances as queries, and the remaining instances are used as the training set. 
The retrieval performance is assessed using two frequently used evaluation metrics, i.e., mean average precision (mAP) and the precision-recall curve (PR curve). We compared the performance of VDeH with the three types of state-of-the-art methods, i.e., (1) thresholding-based methods: spectral hashing (SH), circulant binary embedding (CBE), iterative quantization (ITQ), bilinear projections (BP), and sparse projections (SP); (2) hyperspheres-based method: spherical hashing (SpH); (3) hyperplanes-based methods: density sensitive hashing (DSH), sparse embedding and least variance encoding (SELVE), and refining codes for locality sensitive hashing (rcLSH). The parameters within each hashing method were assigned to default or suggested values by authors. For VDeH, the parameter $\psi$ is searched in $\{2^2,2^3,\dots,2^8\}$. All experiments are executed on a Linux CPU machine: AMD 128-core CPU with each core running at 2 GHz and 1T GB RAM.

\subsection{Results and discussion}

\textbf{Comparing to the state-of-the-art.}
Table~\ref{tab:mAP_main} presents the mAP scores of our proposed VDeH and the competing hashing methods conducted on the benchmark datasets.  VDeH has the best results, compared with all the state-of-the-art methods, across all the tested number of hash bits ranging from 128 bits to 2048 bits. The mAP of VDeH increases consistently as the number of hash bits increases. 
Overall, only DSH demonstrated comparable performance across all datasets, positioning it as VDeH’s closest contender. Among existing methods, rcLSH exhibited the best performance on image datasets. Notably, many existing methods showed a marked decrease in performance on the text datasets as the hash bits increases, except VDeH, SH, and DSH.
This is probably because text data have complex density variations. DSH and SH are adapted by using more hash functions in dense distribution and fewer in sparse distribution; while VDeH naturally adapts because it creates smaller Voronoi cells in dense distribution and larger ones in sparse distribution~\cite{Devroye2017OnTM,ting2018isolation}.
In addition, Figure~\ref{fig:PRcurves} provides the PR curves of VDeH and the competing hashing methods, when the number of hash bits is 512. We can observe that the VDeH curve (in red) consistently dominates the curves of other methods across all datasets, demonstrating its superior precision and recall scores. For the cases with different hash bits, similar results can be observed (not presented due to lack of space).

\begin{table*}[t!]
\centering
\caption{The mAP scores on four datasets with different number of hash bits. The highest score in each row is marked with bold font.}
\scalebox{0.82}{
\begin{tabular}{c|c|cccccccccc}
\hline
      \textbf{Dataset}               & \textbf{  bits  }   & \textbf{  
 VDeH  }             & \textbf{rcLSH}    & \textbf{  SH  }          & \textbf{  CBE  }      & \textbf{  ITQ  }      & \textbf{  \, BP \,  }     & \textbf{  SP  }      & \textbf{  DSH  }     & \textbf{SELVE}   & \textbf{  SpH  }  \\ \hline
                                     & 128    & {\textbf{0.366}}  & 0.283     & 0.117       & 0.315     & {0.355}     & 0.344    & 0.352    & 0.318    & 0.199    & 0.242          \\
                                     & 256    & {\textbf{0.438}}  & 0.351     & 0.153       & 0.342     & 0.377     & 0.343    & 0.370    & 0.372    & 0.176    & 0.289          \\
                                     & 512    & {\textbf{0.468}}  & 0.392     & 0.193       & 0.396     & 0.389     & 0.339    & 0.397    & 0.401    & 0.167    & 0.304          \\
                                     & 1024   & {\textbf{0.501}}  & 0.412     & 0.201       & 0.401     & 0.398     & 0.332    & 0.407    & 0.429    & 0.169    & 0.314          \\
\multirow{-5}{*}{\rotatebox{90}{ \textbf{CIFAR-10} }}       & 2048   & {\textbf{0.525}}  & 0.432     & 0.214       & 0.407     & 0.401     & 0.335    & 0.426    & 0.457    & 0.169    & 0.321          \\ \hline
                                     & 128    & {\textbf{0.664}}  & 0.582     & 0.338       & 0.617     & 0.644     & 0.645    & 0.604    & 0.501    & 0.430    & 0.222          \\
                                     & 256    & {\textbf{0.714}}  & 0.631     & 0.454       & 0.644     & 0.650     & 0.641    & 0.614    & 0.580    & 0.399    & 0.223          \\
                                     & 512    & {\textbf{0.763}}  & 0.651     & 0.501       & 0.656     & 0.654     & 0.634    & 0.638    & 0.618    & 0.402    & 0.232          \\
                                     & 1024   & {\textbf{0.774}}  & 0.659     & 0.514       & 0.657     & 0.657     & 0.642    & 0.638    & 0.610    & 0.399    & 0.233          \\
\multirow{-5}{*}{\rotatebox{90}{\textbf{GIST}}}               & 2048   & {\textbf{0.793}}  & 0.687     & 0.532       & 0.662     & 0.663     & 0.644    & 0.641    & 0.621    & 0.398    & 0.242          \\ \hline
                                     & 128    & {\textbf{0.116}}  & 0.017     & 0.033       & 0.022     & 0.023     & 0.023    & 0.022    & 0.016    & 0.003    & 0.100          \\
                                     & 256    & {\textbf{0.165}}  & 0.026     & 0.055       & 0.029     & 0.029     & 0.029    & 0.027    & 0.033    & 0.004    & 0.110          \\
                                     & 512    & {\textbf{0.340}}  & 0.028     & 0.160       & 0.026     & 0.030     & 0.026    & 0.029    & 0.185    & 0.004    & 0.146          \\
                                     & 1024   & {\textbf{0.348}}  & 0.103     & 0.295       & 0.108     & 0.104     & 0.008    & 0.106    & 0.339    & 0.006    & 0.088           \\
\multirow{-5}{*}{\rotatebox{90}{\textbf{Nytimes}}}            & 2048   & \textbf{0.376}    & 0.113     & 0.310       & 0.100     & 0.097     & 0.008    & 0.100    & 0.326    & 0.008    & 0.089                        \\ \hline
                                     & 128    & {\textbf{0.457}}  & 0.235     & 0.375       & 0.255     & 0.261     & 0.254    & 0.258    & 0.368    & 0.203    & 0.256          \\
                                     & 256    & {\textbf{0.603}}  & 0.271     & 0.512       & 0.286     & 0.293     & 0.295    & 0.304    & 0.494    & 0.226    & 0.296          \\
                                     & 512    & {\textbf{0.607}}  & 0.308     & 0.571       & 0.301     & 0.298     & 0.346    & 0.341    & 0.437    & 0.279    & 0.346          \\
                                     & 1024   & \textbf{0.615}    & 0.307     & 0.569       & 0.300     & 0.302     & 0.351    & 0.357    & 0.415    & 0.289    & 0.372                       \\
\multirow{-5}{*}{\rotatebox{90}{\textbf{Kosarak}}}            & 2048   & \textbf{0.638}    & 0.313     & 0.579       & 0.309     & 0.312     & 0.359    & 0.369    & 0.412    & 0.300    & 0.381                       \\ \hline
% \multicolumn{2}{|c|}{Avg. mAP}                & {\textbf{53.1}}  & 42.3     & 31.2       & 42.6     & 42.5     & 36.0    & 43.7    & 40.1    & 22.7    & 32.9          \\ \hline
\end{tabular}}
\label{tab:mAP_main}
\end{table*}

\begin{figure*}[t!]
    \centering
    \includegraphics[width=0.96\textwidth]{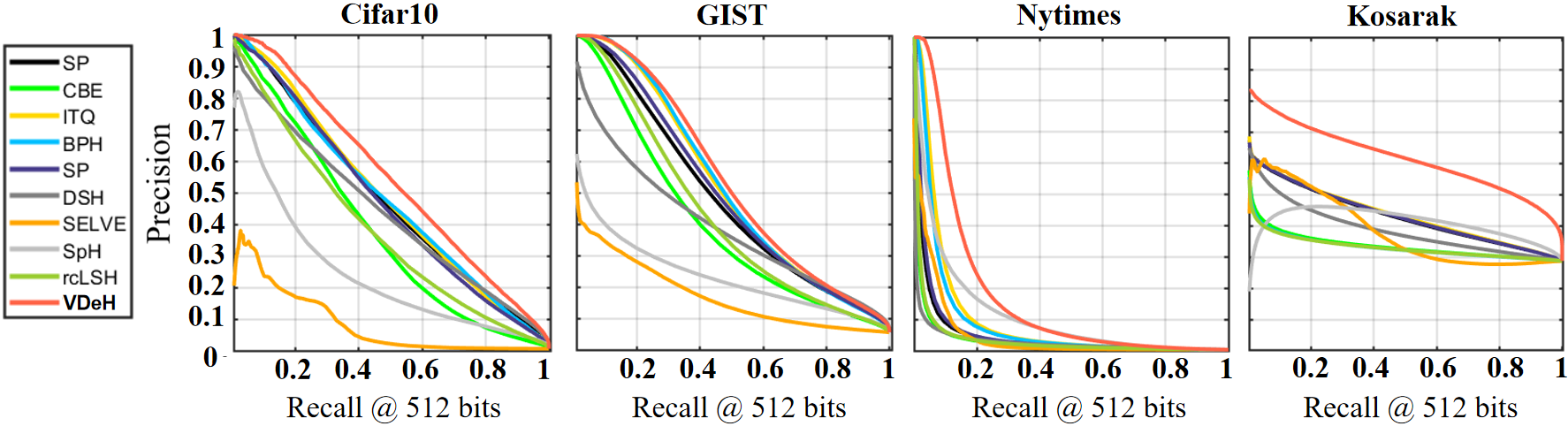}
    \caption{PR curves on four datasets with the 512-bit binary codes}
    \label{fig:PRcurves}
\end{figure*}

\begin{figure}[t!]
    \centering
    % Subfigure (a)
    \begin{subfigure}[b]{0.36\textwidth} % Adjust the width as needed
        \centering
        \includegraphics[width=1\textwidth]{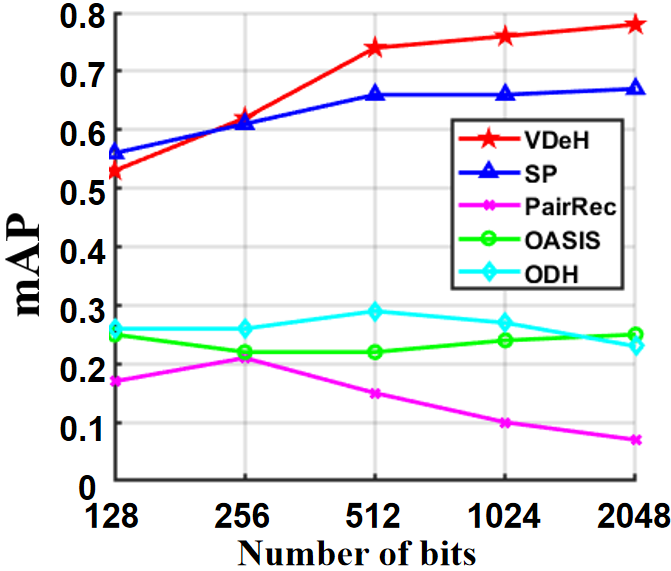} % Replace with your image
        \caption{Cifar10} % Caption for subfigure (a)
        \label{fig:subfig-a} % Label for referencing
    \end{subfigure}
    \hspace{0.05\textwidth} % Adjust the space as needed
    % \hfill % Optional: Adds space between the two subfigures
    % Subfigure (b)
    \begin{subfigure}[b]{0.36\textwidth} % Adjust the width as needed
        \centering
        \includegraphics[width=1\textwidth]{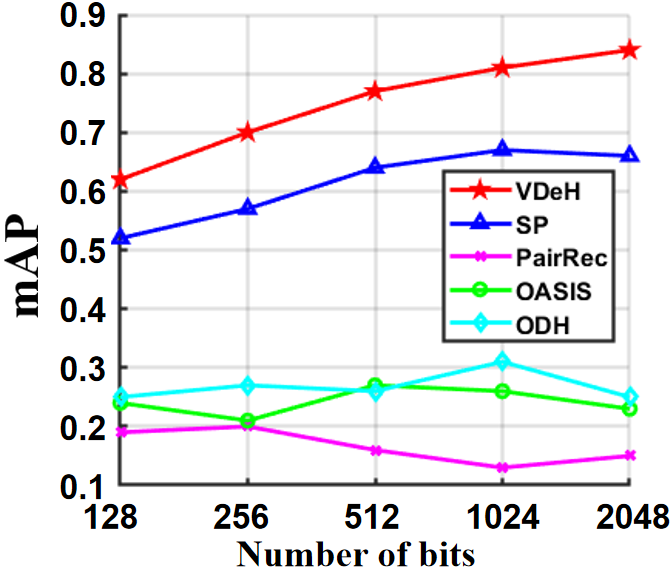} % Replace with your image
        \caption{GIST} % Caption for subfigure (b)
        \label{fig:subfig-b} % Label for referencing
    \end{subfigure}
    \caption{The mAP scores comparison for five methods on image datasets.}
    \label{learning_com}
\end{figure}

\noindent\textbf{Comparing to the deep hashing methods.} Deep hashing is known for its high accuracy in semantic-based image retrieval~\cite{luo2023survey}. However, deep hashing methods cannot accurately retrieve input distance similarities, even after extracting semantic information from images. To show this, we use CIFAR-$10$ and GIST datasets, which possess semantic information in images. Initially, we utilized ResNet18 to extract embeddings from these two datasets. Since the Euclidean distance between the obtained embeddings can reflect the semantic information between images, we used the Euclidean distances derived from these embeddings as ground truth to test the 
performance of VDeH and SP, as well as three deep hashing methods PairRec~\cite{Hansen2020UnsupervisedSHDeep}, OASIS~\cite{Wu2022OnlineEnhancedDeep} and ODH~\cite{he2024one}. The comparison results are shown in Figure~\ref{learning_com}. We can observe that typical L2H methods can effectively preserve the Euclidean distances between instances, whereas deep methods lack this capability. This is due to the fact that deep methods incorporate labels as part of their loss function during training, aiming to map instances with the same label to close hash codes, rather than being designed to preserve the input distance similarities between instances. 

\begin{figure}[t!]
    \centering
    \includegraphics[width=0.66\textwidth]{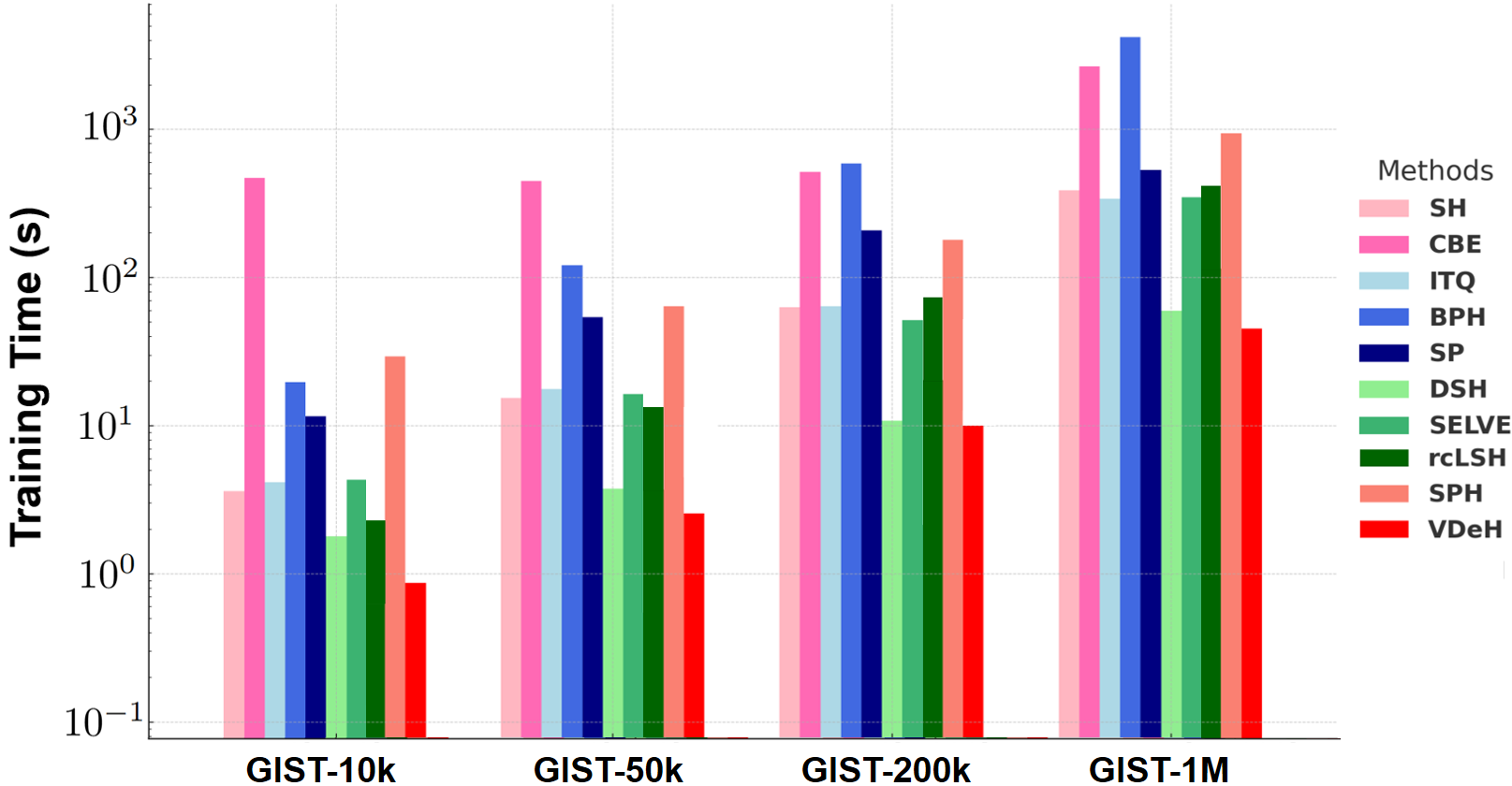}
    \caption{Training time comparison on GIST dataset with the 512-bit binary codes. 
    }
    \label{fig:time_g}
\end{figure}

\noindent\textbf{Training time comparison.} Unlike existing hashing methods that rely on optimization for learning hash functions, VDeH's training process is straightforward, requiring only the random sampling of a given number of points from the training data, each corresponding to a nearest-neighbor-based hash function (as shown in Eq.~\ref{eqhash}). Note that all methods require the transformation of training data into binary codes via hash functions. We tested the training time of various methods on the GIST dataset, shown in Figure~\ref{fig:time_g}. VDeH took the shortest time for every training data size, from 10k to 1M, demonstrating its efficiency superiority over other methods. This is because VDeH needs no learning, but all other existing methods must perform learning.

\begin{figure}[t!]
    \centering
    % Subfigure (a)
    \begin{subfigure}[b]{0.36\textwidth} % Adjust the width as needed
        \centering
        \includegraphics[width=1\textwidth]{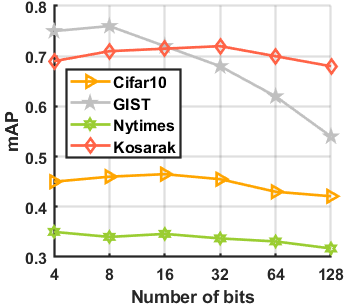} % Replace with your image
        \caption{512-bit binary codes} % Caption for subfigure (a)
        \label{fig:subfig-a} % Label for referencing
    \end{subfigure}
    \hspace{0.05\textwidth} % Adjust the space as needed
    % \hfill % Optional: Adds space between the two subfigures
    % Subfigure (b)
    \begin{subfigure}[b]{0.36\textwidth} % Adjust the width as needed
        \centering
        \includegraphics[width=1\textwidth]{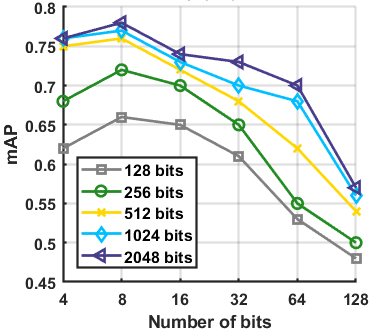} % Replace with your image
        \caption{GIST dataset} % Caption for subfigure (b)
        \label{fig:subfig-b} % Label for referencing
    \end{subfigure}
    \caption{The impact of Parameter $\psi$: (a) testing on different datasets with 512-bit binary codes; (b) Testing on GIST dataset with different number of hash bits.}
    \label{para_com}
\end{figure}

\noindent\textbf{Effects of hyper-parameter $\psi$.}
Figure~\ref{para_com} illustrates the impact of the parameter $\psi$ on the performance of the proposed VDeH method across different datasets and binary code lengths. In subfigure (a), we observe that only the GIST dataset shows a degree of sensitivity to $\psi$, while Cifar10, Nytimes, and Kosarak show relatively stable mAP scores across the different $\psi$ values. Subfigure (b) further explores this by examining VDeH's mAP scores on the GIST dataset with varying binary code lengths. We can observe that the optimal $\psi$ value remains consistent regardless of the number of hash bits, ranging from 128 to 2048 bits. This suggests that the optimal selection of $\psi$ is intrinsic to the characteristics of the specific dataset, rather than the code length.

\section{Discussion}

We acknowledge that the integration of Voronoi diagrams with hashing techniques is not an entirely novel concept~\cite{ajani2013efficient,loi2013vlsh}. However, our proposed VDeH distinguishes itself from existing works through its unique mechanism for generating data-dependent hash functions and its redefinition of the hashing process. 
Prior approaches, such as VLSH by Loi \textit{et al.}~\cite{loi2013vlsh}, use Voronoi regions to primarily localize and adapt standard Locality Sensitive Hashing (LSH); its core hashing still relies on LSH's random projection-based mechanism within these regions. Similarly, the work by Ajani \& Wanjari combines Voronoi clustering with hash indexing, where hash indexing serves to optimize their $k$-means clustering algorithm for uncertain data rather than generating binary codes for similarity search. 
In contrast, VDeH directly derives hash functions from the partitions of Voronoi diagrams themselves. The goal of VDeH is to generate highly efficient binary representations for large-scale datasets, thereby enabling rapid similarity search and retrieval.
The main novelty of VDeH lies in its utilization of Voronoi diagrams as the core data-dependent hash function generator. It proposes a method that can achieve key hashing properties, including full space coverage, entropy maximization, and bit independence, without resorting to complex learning procedures, notably through its specifically defined encoded hashing mechanism.

\section{Conclusion} 

We introduce VDeH, a novel no-learning data-dependent method for binary hashing. VDeH distinguishes itself from L2H methods in three aspects. First, VDeH employs the unique space partitioning of Voronoi diagrams, leveraging its three previously concealed properties that match perfectly those required for hashing. Second, VDeH is an implementation of a new definition of hashing which equates the probability of some hashed condition to a similarity function. The definition enables Voronoi diagrams, already used to compute the similarity of Isolation Kernel, to construct hash functions easily without much effort. No existing L2H methods have used a similar definition as far as we know. Third, the integration of encoded hashing enables VDeH to generate mutually independent binary bits, which could not be utilized by existing methods because they have already employed independent hash functions.
Our experiments show that VDeH exhibits superior performance with lower computational cost compared to the state-of-the-art methods.

\begin{credits}
\subsubsection{\ackname} We thank the anonymous reviewers for their valuable comments. This work was supported in part by the National Natural Science Foundation of China (Grant No. 92470116).
\end{credits}

\end{document}